\documentclass[10pt,letterpaper]{article}
\usepackage[top=0.85in,left=2.75in,footskip=0.75in]{geometry}

\usepackage{hyperref,graphicx,amsmath,amssymb,amsthm,bm}

\usepackage{changepage}

\usepackage[utf8x]{inputenc}

\usepackage{textcomp,marvosym}

\usepackage{cite}

\usepackage{nameref,hyperref}

\usepackage[right]{lineno}

\usepackage{microtype}
\DisableLigatures[f]{encoding = *, family = * }

\usepackage[table]{xcolor}

\usepackage{array}

\newcolumntype{+}{!{\vrule width 2pt}}

\newlength\savedwidth

\raggedright
\usepackage[aboveskip=1pt,labelfont=bf,labelsep=period,justification=raggedright,singlelinecheck=off]{caption}

\makeatletter
\renewcommand{\@biblabel}[1]{\quad#1.}
\makeatother

\usepackage{lastpage,fancyhdr,graphicx}
\usepackage{epstopdf}
\fancyheadoffset[L]{2.25in}
\fancyfootoffset[L]{2.25in}
\DeclareMathOperator{\spn}{span}
\DeclareMathOperator{\rank}{rank}

\newtheorem{definition}{\sffamily\bfseries Definition}
\newtheorem{proposition}{\sffamily\bfseries Proposition}
\newtheorem{lemma}{\sffamily\bfseries Lemma}

\begin{document}
\vspace*{0.2in}

% Title must be 250 characters or less.
\begin{flushleft}
{\Large
\textbf\newline{Linear algebraic structure of zero-determinant strategies in repeated games} % Please use "sentence case" for title and headings (capitalize only the first word in a title (or heading), the first word in a subtitle (or subheading), and any proper nouns).
}
\newline
% Insert author names, affiliations and corresponding author email (do not include titles, positions, or degrees).
\\
Masahiko Ueda\textsuperscript{1*},
Toshiyuki Tanaka\textsuperscript{1}
\\
\bigskip
\textbf{1} Department of Systems Science, Graduate School of Informatics, Kyoto University, Kyoto 606-8501, Japan
\\
\bigskip

% Insert additional author notes using the symbols described below. Insert symbol callouts after author names as necessary.
% 
% Remove or comment out the author notes below if they aren't used.
%
% Primary Equal Contribution Note
%\Yinyang These authors contributed equally to this work.

% Additional Equal Contribution Note
% Also use this double-dagger symbol for special authorship notes, such as senior authorship.
%\ddag These authors also contributed equally to this work.

% Current address notes
%\textcurrency Current Address: Dept/Program/Center, Institution Name, City, State, Country % change symbol to "\textcurrency a" if more than one current address note
% \textcurrency b Insert second current address 
% \textcurrency c Insert third current address

% Deceased author note
%\dag Deceased

% Group/Consortium Author Note
%\textpilcrow Membership list can be found in the Acknowledgments section.

% Use the asterisk to denote corresponding authorship and provide email address in note below.
* ueda.masahiko.5r@kyoto-u.ac.jp

\end{flushleft}
% Please keep the abstract below 300 words
\section*{Abstract}
Zero-determinant (ZD) strategies, a recently found novel class of strategies in repeated games, has attracted much attention in evolutionary game theory. 
A ZD strategy unilaterally enforces a linear relation between average payoffs of players.
Although existence and evolutional stability of ZD strategies have been studied in simple games, their mathematical properties have not been well-known yet.
For example, what happens when more than one players employ ZD strategies have not been clarified.
In this paper, we provide a general framework for investigating situations where more than one players employ ZD strategies in terms of linear algebra.
First, we theoretically prove that a set of linear relations of average payoffs enforced by ZD strategies always has solutions, which implies that incompatible linear relations are impossible.
Second, we prove that linear payoff relations are independent of each other under some conditions.
These results hold for general games with public monitoring including perfect-monitoring games.
Furthermore, we provide a simple example of a two-player game in which one player can simultaneously enforce two linear relations, that is, simultaneously control her and her opponent's average payoffs.
All of these results elucidate general mathematical properties of ZD strategies.

% Please keep the Author Summary between 150 and 200 words
% Use first person. PLOS ONE authors please skip this step. 
% Author Summary not valid for PLOS ONE submissions.   
%\section*{Author summary}

%\linenumbers

% Use "Eq" instead of "Equation" for equation citations.
\section*{Introduction}
Game theory is a powerful framework explaining rational behaviors of human beings~\cite{FudTir1991} and evolutionary behaviors of biological systems~\cite{SmiPri1973,Now2006}. 
In a simple example of prisoner's dilemma game, mutual defection is realized as a result of rational thought, even if mutual cooperation is more favorable.
On the other hand, when the game is repeated infinite times, cooperation can be realized if players are far-sighted, which is confirmed as folk theorem.
Axelrod's famous tournaments on infinitely repeated prisoner's dilemma game~\cite{AxeHam1981,Axe1984} also showed that cooperative but retaliating strategy, called the tit-for-tat strategy, is successful in the setting of infinitely repeated game.

Recently, in repeated games with perfect monitoring, a novel class of strategies, called zero-determinant (ZD) strategy, was discovered~\cite{PreDys2012}.
Surprisingly, ZD strategy unilaterally enforces a linear relation between average payoffs of players.
A strategy which unilaterally sets her opponent's average payoff (equalizer strategy) is one example.
Another example is extortionate strategy in which the player can earn more average payoff than her opponent.
ZD strategies contain the well-known tit-for-tat strategy as a special example.
After the pioneering work of Press and Dyson, stability of ZD strategies has been studied in the context of evolutionary game theory~\cite{HNS2013,AdaHin2013,StePlo2013,HNT2013,StePlo2012,SzoPer2014}, and it was found that some kind of ZD strategies, called generous ZD strategies, can stably exist.
Performance of ZD strategies has also been studied in human experiments~\cite{HRM2016,WZLZX2016}.
Although ZD strategy was originally formulated in two-player two-action (iterated prisoner's dilemma) games, ZD strategy was extended to multi-player two-action (iterated social dilemma) games~\cite{HWTN2014,PHRT2015}, two-player multi-action games~\cite{Guo2014,McAHau2016}, and multi-player multi-action games~\cite{HDND2016}.
In addition, ZD strategy was extended to two-player two-action noisy games~\cite{HRZ2015,MamIch2019}, which is one example of the repeated games with imperfect monitoring.
Furthermore, besides these fundamental theoretical studies, ZD strategies are also applied to resource sharing in wireless networks~\cite{DKL2014,ZNSJH2016}.
See Ref.~\cite{HCN2018} for a review of ZD strategies in the context of direct reciprocity.

The contributions of this paper are four-fold. 
First, we extend ZD strategy for general multi-player multi-action repeated games with public monitoring, where players know the structure of games (players, sets of actions of all players, and payoffs of all players) but cannot observe actions of other players.
A typical example of such situation is auction.
In a sealed-bid auction, a player cannot know actions (bids) of other players, but only knows the result of the game (whether she is the winner or not).
Second, we prove, in terms of a linear-algebraic argument,
that linear payoff relations enforced by players with ZD strategies are consistent, that is, always have solutions.
Third, we introduce the notion of independence of ZD strategies,
and prove, again in terms of a linear-algebraic argument,
that linear payoff relations enforced by players with ZD strategies
are independent under a general condition. 
Fourth, as an application of linear algebraic formulation, we provide a simple example of a two-player game in which one player can simultaneously enforce two linear relations.
This means that she can simultaneously control her and her opponent's average payoffs, which has never been reported in the context of ZD strategies.
All of these results develop deeper understanding of mathematical properties of ZD strategies in general games.

We remark on discounting.
In standard repeated games, discounting of future payoffs is considered by introducing a discounting factor $\delta\leq 1$~\cite{FudTir1991}.
In the original work on ZD strategy by Press and Dyson, only the case without discounting (i.e., $\delta=1$) was investigated~\cite{PreDys2012}.
After their work, ZD strategy was extended to $\delta<1$ case~\cite{HTS2015,McAHau2016,IchMas2018}.
In this paper, we consider only the non-discounting case $\delta=1$.

\section*{Setup}
We consider an $N$-player multi-action repeated game,
in which player $n\in\{1,\cdots,N\}$ has $M_n$ possible actions,
where $M_n$ is a positive integer. 
Let $\bm{\sigma}\equiv (\sigma_1,\cdots,\sigma_N)
\in\Sigma\equiv \prod_{n=1}^N\{1,\cdots,M_n\}$ denote
a state of the game, which is the combination
of the actions taken by the $N$ players.
Let $M\equiv \prod_{n=1}^NM_n$ be the size of the state space $\Sigma$. 
We assume that player $n$ decides the next action
stochastically according to her own previous action $\sigma^\prime_n$ and common information $\tau \in B$
with the conditional probability $\hat{T}_n\left( \sigma_n | \sigma^\prime_n, \tau \right)$, where $B$ is some set.
We also define the conditional probability that common information $\tau$ arises when actions of players in the preceding round are $\bm{\sigma}'$ by $W\left( \tau | \bm{\sigma}' \right)$.
(An example of $\tau$ is the winner in each round; see \nameref{S1_Text}.)
Then the sequence of states of the repeated game
forms a Markov chain
\begin{equation}
 P\left( \bm{\sigma}, t+1 \right) = \sum_{\bm{\sigma}'} T\left( \bm{\sigma} | \bm{\sigma}' \right) P\left( \bm{\sigma}', t \right)
 \label{eq:mc}
\end{equation}
with the transition probability
\begin{equation}
 T\left( \bm{\sigma} | \bm{\sigma}' \right) \equiv \sum_\tau W\left( \tau | \bm{\sigma}' \right) \prod_{n=1}^N \hat{T}_n\left( \sigma_n | \sigma^\prime_n, \tau \right),
\end{equation}
where $P(\bm{\sigma},t)$ denotes the state distribution
at time $t$. 
We assume that all players know the function $W\left( \tau | \bm{\sigma}' \right)$ but cannot directly observe $\bm{\sigma}'$.
When $B=\Sigma$ and $W\left( \tau | \bm{\sigma}' \right)=\delta_{\tau, \bm{\sigma}'}$,
the above formulation reduces to that of perfect monitoring games.
Otherwise, it represents games with public monitoring, where players cannot directly observe actions of other players.
The model treated here can therefore be regarded as
an extension of repeated games with perfect monitoring
to those with imperfect monitoring,
and the extension includes the former as a special case. 

For each state $\bm{\sigma}$, a payoff of player $n$ is defined as $s_n\left( \bm{\sigma} \right)$.
Let $\bm{s}_n\equiv (s_n(\bm{\sigma}^\prime))_{\bm{\sigma}^\prime\in\Sigma}$ be
the $M$-dimensional vector representing the payoffs
of player $n$, which we call the payoff vector of player $n$.
It should be noted that in the following analysis
we do not assume the payoffs to be symmetric,
unless otherwise stated.

\section*{Results}
\subsection*{Zero-determinant strategies}
Because a discounting factor $\delta$ is one, the payoffs of players are the average payoffs with respect to the stationary distribution of the Markov chain.
Let $P^{(\mathrm{s})}\left( \bm{\sigma} \right)$ denote
the stationary distribution, which may depend on the initial condition when the Markov chain is not irreducible.
It satisfies
\begin{equation}
 P^{(\mathrm{s})}\left( \bm{\sigma} \right) = \sum_{\bm{\sigma}'} T\left( \bm{\sigma} | \bm{\sigma}' \right) P^{(\mathrm{s})}\left( \bm{\sigma}' \right).
 \label{eq:steady}
\end{equation}
Taking summation of both sides of Eq.~\eqref{eq:steady}
with respect to $\bm{\sigma}_{-n}\equiv\bm{\sigma}\backslash\sigma_n$
with an arbitrary $n$, we obtain 
\begin{equation}
  \label{eq:Akin}
 0 = \sum_{\bm{\sigma}'} \left[ T_n\left( \sigma_n | \bm{\sigma}' \right) - \delta_{\sigma_n, \sigma^\prime_n} \right] P^{(\mathrm{s})}\left( \bm{\sigma}' \right),
\end{equation}
where we have defined
\begin{equation}
 T_n\left( \sigma_n | \bm{\sigma}' \right) \equiv \sum_\tau W\left( \tau | \bm{\sigma}' \right) \hat{T}_n\left( \sigma_n | \sigma^\prime_n, \tau \right).
 \label{eq:transition_n}
\end{equation}
Regarding $\delta_{\sigma_n,\sigma_n'}$ as 
representing the strategy ``Repeat'', 
where player $n$ repeats the previous action with probability one, 
one can readily see that Eq.~\eqref{eq:Akin} is an extension
of Akin's lemma~\cite{Aki2012,HWTN2014,Aki2015,McAHau2016},
relating a player's strategy with the stationary distribution,
to the multi-player multi-action public-monitoring case. 
Letting
\begin{equation}
  \label{eq:def_tTn}
 \tilde{T}_n\left( \sigma_n | \bm{\sigma}' \right) \equiv T_n\left( \sigma_n | \bm{\sigma}' \right) - \delta_{\sigma_n, \sigma^\prime_n},
\end{equation}
Eq.~\eqref{eq:Akin} means that the average of $\tilde{T}_n\left( \sigma_n | \bm{\sigma}' \right)$ with respect to the stationary distribution
is zero for any $n$ and $\sigma_n$.
We remark that all players are assumed to know the functional form of $W\left( \tau | \bm{\sigma}' \right)$, and that $\hat{T}_n\left( \sigma_n | \sigma^\prime_n, \tau \right)$,
and thus $T_n(\sigma_n|\bm{\sigma}')$ as well,
are solely under control of player $n$.
Because of the normalization condition $\sum_{\sigma_n=1}^{M_n} T_n\left( \sigma_n | \bm{\sigma}' \right) = 1$, the relation
\begin{equation}
 \sum_{\sigma_n=1}^{M_n} \tilde{T}_n\left( \sigma_n | \bm{\sigma}' \right) = 0
 \label{eq:tt_norm}
\end{equation}
holds.

Let $\tilde{\bm{T}}_n(\sigma_n)\equiv (\tilde{T}_n(\sigma_n|\bm{\sigma^\prime}))_{\bm{\sigma^\prime}\in\Sigma}$,
which we call the strategy vector of player $n$
associated with action $\sigma_n$.
(Another name for $\tilde{\bm{T}}_n(\sigma_n)$
is the Press-Dyson vector~\cite{Aki2012}.) 
A strategy of player $n$ is represented as an $M\times M_n$ matrix
$\mathcal{T}_n\equiv (\tilde{\bm{T}}_n(1),\cdots,\tilde{\bm{T}}_n(M_n))$
composed of the strategy vectors for her actions $\sigma_n\in\{1,\ldots,M_n\}$. 
For a matrix $A$, 
let $\spn A$ be the subspace
spanned by the column vectors of $A$.
Let $\mathbf{0}_m$ and $\mathbf{1}_m$ denote the $m$-dimensional zero vector
and the $m$-dimensional vector of all ones, respectively. 
From Eq.~\eqref{eq:tt_norm}, one has
\begin{equation}
  \label{eq:null}
  \mathcal{T}_n\mathbf{1}_{M_n}
  =\sum_{\sigma_n=1}^{M_n} \tilde{\bm{T}}_n(\sigma_n)
  =\mathbf{0}_M
\end{equation}
for any player $n$,
implying that the dimension of $\spn\mathcal{T}_n$ is at most $(M_n-1)$.

Let $\bm{\rho}\equiv (P^{(\mathrm{s})}(\bm{\sigma}))_{\bm{\sigma}\in\Sigma}$
be the vector representation of the stationary distribution
$P^{(\mathrm{s})}(\bm{\sigma})$. 
When player $n$ chooses a strategy $\mathcal{T}_n$, 
for any vector $\bm{v}\in\spn\mathcal{T}_n$,
one has $\bm{\rho}^\mathsf{T}\bm{v}=0$ due to Eq.~\eqref{eq:Akin}.
In other words,
the expectation of $\bm{v}$ with respect to the stationary distribution
$\bm{\rho}$ vanishes.

Let $\mathcal{S}\equiv (\mathbf{1}_M,\bm{s}_1,\cdots,\bm{s}_N)$
and $V_n\equiv \spn\mathcal{T}_n\cap\spn\mathcal{S}$.
The following definition is an extension of the notion of the ZD strategy \cite{PreDys2012,Aki2012}
to multi-player multi-action public-monitoring games.

\begin{definition}
  A \emph{zero-determinant (ZD) strategy} is defined as a strategy
  $\mathcal{T}_n$ for which $\dim V_n\ge1$ holds. 
\end{definition}

To see that this is indeed an extended definition of the ZD strategy,
note that any vector $\bm{u}\in\spn\mathcal{S}$ is represented as
$\bm{u}=\mathcal{S}\bm{\alpha}$, where
$\bm{\alpha}\equiv (\alpha_0,\alpha_1,\cdots,\alpha_N)^\mathsf{T}$
is the coefficient vector.
Let $\bm{e}\equiv (1,e_1,\cdots,e_N)^\mathsf{T}=\mathcal{S}^\mathsf{T}\bm{\rho}$
be the vector with element $e_n$ equal to the expected payoff
$e_n\equiv \left\langle s_n(\bm{\sigma}) \right\rangle_\mathrm{s}$ of player $n$ in the steady state.
When player $n$ employs a ZD strategy,
it amounts to enforcing linear relations $\bm{e}^\mathsf{T}\bm{\alpha}=\bm{\rho}^\mathsf{T} \mathcal{S}\bm{\alpha} =0$
on $\bm{e}$ with $\bm{\alpha}$ satisfying
$\mathcal{S}\bm{\alpha}\in V_n$.

\subsection*{Consistency}
A question naturally arises:
When more than one of the players employ ZD strategies,
are they ``consistent'', that is, do linear payoff relations enforced by the players always have solutions?
For example, in a two-player game, when player $1$ enforces $\sum_{n=1}^2 \alpha_n e_n = \gamma$ by a ZD strategy and player $2$ enforces $\sum_{n=1}^2 \alpha^\prime_n e_n = \gamma^\prime$ by a ZD strategy, do the simultaneous equations of $(e_1, e_2)$ have a solution?
Let $N'$ be the set of players who employ ZD strategies.
The set 
$E\equiv \{\bm{e}\in\{1\}\times\mathbb{R}^N:\bm{e}^\mathsf{T}\bm{\alpha}=0,\forall\bm{\alpha},
\mathcal{S}\bm{\alpha}\in\spn(V_n)_{n\in N'}\}$ consists of
all combinations of the expected payoffs that satisfy the
enforced linear relations by the players in $N'$. 
If $E$ is empty, then it implies that the set of ZD strategies
is inconsistent in the sense that there is no valid solution
of the linear relations enforced by the players.

\begin{definition}
  ZD strategies are said to be \emph{consistent} when $E$ is not empty.
\end{definition}

In the multi-player setting,
one may regard $N'$ as a variant of a ZD strategy alliance~\cite{HWTN2014},
where the players in $N'$ agree to coordinate
on the linear relations to be enforced on the expected payoffs.
The above question then amounts to asking whether it is possible
for a player to serve as a counteracting agent
who participates in the ZD strategy alliance
with a hidden intention to invalidate it 
by adopting a ZD strategy that is inconsistent with others. 

The following proposition is the first main result of this paper.
\begin{proposition}
\label{prop:ZDconsistent}
Any set of ZD strategies is consistent.
\end{proposition}

\begin{proof}
We first note that the following property holds for strategy vectors, whose proof is given in Methods.
\begin{lemma}
\label{lemma:consistent}
Let $\mathcal{T}=(\mathcal{T}_1,\cdots,\mathcal{T}_N)$. 
Then $\mathbf{1}_M\not\in\spn\mathcal{T}$.
\end{lemma}

%We return to the proof of Proposition~\ref{prop:ZDconsistent}. 
For any set $\spn(V_n)_{n\in N'}$ of ZD strategies, 
let $K$ be the dimension of $\spn(V_n)_{n\in N'}$,
and let $\bm{u}_1=\mathcal{S}\bm{\alpha}_1,\cdots,\bm{u}_K=\mathcal{S}\bm{\alpha}_K$ be a basis of $\spn(V_n)_{n\in N'}$.
The expected payoff vector $\bm{e}=(1,\bar{\bm{e}}^\mathsf{T})^\mathsf{T}$ should be
given by a non-zero solution 
of the linear equation $\bar{\bm{e}}^\mathsf{T}\bar{A}+\bm{b}^\mathsf{T}=\mathbf{0}_K^\mathsf{T}$
in $\bar{\bm{e}}$, where we define $A$, $\bm{b}$, and $\bar{A}$ as
\begin{equation}
  \label{eq:A}
  A=\left(\begin{array}{c}
    \bm{b}^\mathsf{T}\\\strut\bar{A}\end{array}\right)
  \equiv (\bm{\alpha}_1,
    \bm{\alpha}_2,
    \cdots,
    \bm{\alpha}_K).
\end{equation}
One has
\begin{equation}
  \mathcal{S}A=(
    \bm{u}_1,\bm{u}_2,\cdots,\bm{u}_K)
  =\mathbf{1}_M\bm{b}^\mathsf{T}
  +\bar{\mathcal{S}}\bar{A},
\end{equation}
where $\bar{\mathcal{S}}\equiv (\bm{s}_1,\cdots,\bm{s}_N)$. 

The Rouch\'{e}-Capelli theorem~\cite{ShaRem2012}
tells us that $\rank\bar{A}=\rank A$ is 
a necessary and sufficient condition
for the linear equation $\bar{\bm{e}}^\mathsf{T}\bar{A}+\bm{b}^\mathsf{T}=\mathbf{0}_K^\mathsf{T}$
in $\bar{\bm{e}}$ to have a solution, that is, 
for $\spn(V_n)_{n\in N'}$ to be consistent (because $A$ is augmented matrix).
An equivalent expression of this condition is 
that there is no vector $\bm{c}\in\mathbb{R}^K$
such that $\bar{A}\bm{c}=\mathbf{0}_N$
and $\bm{b}^\mathsf{T}\bm{c}\not=0$ hold (which ensures that there is no elementary operations which make the rank of augmented matrix larger than that of the original matrix). 
Assume to the contrary that there exist $\bm{c}\in\mathbb{R}^K$
such that $\bar{A}\bm{c}=\mathbf{0}_N$ and $\bm{b}^\mathsf{T}\bm{c}\not=0$ hold.
One would then have 
\begin{equation}
  \mathcal{S}A\bm{c}
  =\mathbf{1}_M\bm{b}^\mathsf{T}\bm{c}+\bar{\mathcal{S}}\bar{A}\bm{c}
  =(\bm{b}^\mathsf{T}\bm{c})\mathbf{1}_M.
\end{equation}
On the other hand, $\mathcal{S}A\bm{c}=\sum_{k=1}^Kc_k\bm{u}_k$
is a linear combination of
$\bm{u}_1,\cdots,\bm{u}_K\in\spn(V_n)_{n\in N'}\subset \spn \mathcal{T}$,
so that Lemma~\ref{lemma:consistent} states that it should be zero
if it is proportional to $\mathbf{1}_M$, 
leading to contradiction.
\end{proof}

Proposition~\ref{prop:ZDconsistent} states that
it is impossible for any player to serve as
a counteracting agent to invalidate ZD strategy alliances.
This statement is quite general in that
it applies to any instance of repeated games
covered by our formulation. 

In Ref.~\cite{HDND2016}, it was shown that every player can have at most one master player, who can play an equalizer strategy on the given player
(that is, controlling the expected payoff of the given player), in multi-player multi-action games.
Indeed, our general result on the absence of inconsistent ZD strategies
(Proposition \ref{prop:ZDconsistent}) 
immediately implies that more than one ZD players cannot simultaneously
control the expected payoff of a player to different values. 
Therefore, our result generalizes their result on equalizer strategy to arbitrary ZD strategies.

Since the dimension of $\spn\mathcal{T}_n$ is at most $(M_n-1)$,
depending on $\mathcal{S}$,
it should be possible for player $n$ with $M_n\ge3$
to adopt a ZD strategy for which $\dim V_n\ge2$ holds.
The dimension of $V_n$ corresponds to
the number of independent linear relations to be enforced
on the expected payoffs of the players,
so that it implies that one player may be able to
enforce multiple independent linear relations.
On the other hand, our result on the absence of inconsistent
ZD strategies implies that for any set $N'$ of ZD players
the dimension of $\spn(V_n)_{n\in N'}$ should be at most $N$,
the number of players, 
since any set of ZD strategies should contain at most $N$
independent linear relations if it is consistent. 
This in turn implies that if the dimension of $\spn(V_n)_{n\in N'}$
is equal to $N$ for a subset $N'$ of players
then players not in $N'$ cannot employ independent ZD strategy any more.

\subsection*{Independence}
Another naturally-arising question would be
regarding independence for a set of ZD strategies, 
which we define as follows: 
\begin{definition}
A set $\{\mathcal{T}_n\}_{n\in N'}$ of ZD strategies
is \emph{independent} if any set $\{\bm{v}_n\}_{n\in N'}$
of non-zero vectors $\bm{v}_n$ in $V_n$ is linearly independent.
Otherwise, $\{\mathcal{T}_n\}_{n\in N'}$ is said to be \emph{dependent}. 
\end{definition}
If a set of ZD strategies is dependent,
then there exists a ZD player whose ZD strategy adds
no linear constraints other than those already imposed
by other ZD players.
One of the simplest example of a dependent set of ZD strategies
is the case where two players enforce exactly the same linear relation
to the expected payoffs.
Our second main result is to show 
that any set of ZD strategies is independent under a general condition.
\begin{proposition}
\label{prop:Indep}
Let $N'$ be a subset of players.
Assume that $\tilde{\bm{T}}_n(\sigma_n)$ does not have
zero elements for any $n\in N'$ and any $\sigma_n\in\{1,\ldots,M_n\}$. 
Then, any set $\{\mathcal{T}_n\}_{n\in N'}$ of ZD strategies
of players in $N'$ is independent.
\end{proposition}
See Methods for the proof.

It should be noted that when $\tilde{\bm{T}}_n(\sigma_n)$ has
zero elements then one might have dependent ZD strategies. 
A simple example can be found in a two-player two-action perfect-monitoring (iterated prisoner's dilemma) game:
Let the payoff vectors $\bm{s}_1$ and $\bm{s}_2$ for players 1 and 2 be
$\bm{s}_1 = \left( R, S, T, P \right)^\mathsf{T}$ and $\bm{s}_2 = \left( R, T, S, P \right)^\mathsf{T}$, with $T\not=S$.
If player 1 adopts the strategy
\begin{equation}
  \tilde{\bm{T}}_1(1)
  = \left( 0, -1, 1, 0 \right)^\mathsf{T}
  = \frac{1}{T-S} \bm{s}_1 - \frac{1}{T-S} \bm{s}_2,
\end{equation}
then it enforces the linear payoff relation $e_1=e_2$.
This strategy is a well-known tit-for-tat strategy~\cite{PreDys2012}.
By symmetry, player 2 can also adopt the same strategy $\tilde{\bm{T}}_2(1)=-\tilde{\bm{T}}_1(1)$, 
implying that these two strategies are indeed dependent.

\subsection*{Simultaneous multiple linear relations by one player}
As mentioned above,
when the number $M_n$ of possible actions for player $n$ is more than two, player $n$ may be able to employ a ZD strategy with $\dim V_n\ge2$
to simultaneously enforce more than one linear relations.
(We note that this is impossible for public goods game~\cite{HWTN2014,PHRT2015} because the number of action for each player is two.)
Such a possibility has never been reported in the context of ZD strategies.
Here, we provide a simple example of such a situation in a two-player three-action symmetric game.

We consider the $3\times 3$ symmetric game
\begin{align}
 \bm{s}_1 &= \left( 0, r_1, 0, r_2, 0, 0, 0, 0, 0 \right)^\mathsf{T} \nonumber \\
 \bm{s}_2 &= \left( 0, r_2, 0, r_1, 0, 0, 0, 0, 0 \right)^\mathsf{T}.
\end{align}
We remark that $\textrm{\boldmath $s$}_1$, $\textrm{\boldmath $s$}_2$, and $\textrm{\boldmath $1$}_9$ are linearly independent when $r_1\neq r_2$ and $r_1\neq -r_2$.
We choose strategies of player $1$ as
\begin{align}
 \bm{T}_1(1) &= \left( 1, 1-p, 1, p^\prime, 0, 0, 0, 0, 0 \right)^\mathsf{T} \nonumber \\
 \bm{T}_1(2) &= \left( 0, q, 0, 1-q^\prime, 1, 1, 0, 0, 0 \right)^\mathsf{T} \nonumber \\
 \bm{T}_1(3) &= \left( 0, p-q, 0, q^\prime-p^\prime, 0, 0, 1, 1, 1 \right)^\mathsf{T}
\end{align}
with $0\leq p\leq 1$, $0\leq q\leq 1$, $0\leq p^\prime \leq 1$, $0\leq q^\prime \leq 1$, $q\leq p$, and $p^\prime \leq q^\prime$.
Then we obtain
\begin{align}
 \frac{q^\prime r_1 + qr_2}{p^\prime q - pq^\prime} \bm{\tilde{T}}_1(1) + \frac{p^\prime r_1 + pr_2}{p^\prime q - pq^\prime} \bm{\tilde{T}}_1(2) &= \bm{s}_1 \\
 \frac{q^\prime r_2 + qr_1}{p^\prime q - pq^\prime} \bm{\tilde{T}}_1(1) + \frac{p^\prime r_2 + pr_1}{p^\prime q - pq^\prime} \bm{\tilde{T}}_1(2) &= \bm{s}_2.
\end{align}
Therefore, player $1$ can simultaneously control average payoffs of both players, $e_1$ and $e_2$, as $e_1 = e_2 = 0$.
Note that $\bm{\sigma}$ with $s_1(\bm{\sigma})=0$ is an absorbing state regardless of the strategy of player $2$ in this case.

In general, when one player simultaneously enforces two linear relations in two-player multi-action symmetric games, only $e_1 = e_2 = C$ is allowed with some $C$.
This is explained as follows: 
Assume that player $1$ can simultaneously enforce $e_1=C_1$ and $e_2=C_2$ with $C_1\neq C_2$ by one ZD strategy.
Because the game is symmetric, player $2$ can also simultaneously enforce $e_1=C_2$ and $e_2=C_1$ independently by one ZD strategy.
This contradicts the consistency of ZD strategies (Proposition \ref{prop:ZDconsistent}).
Therefore, the only possibility is $e_1 = e_2 = C$.

The above argument can be extended straightforwardly
to the multi-player case.
For that purpose, we introduce some notions of symmetric multi-player games.
The following definition of a symmetric multi-player game
is due to von Neumann and Morgenstern~\cite[Section 28]{vonNeumannMorgenstern1953}. 
\begin{definition}
\label{def:symmetric}
A game is \emph{symmetric}
with respect to a permutation $\pi$ on $\{1,\ldots,N\}$
if $M_n=M_{\pi(n)}$ holds for any $n\in\{1,\ldots,N\}$
and if $\pi$ preserves the payoff structure of the game, that is,
\begin{equation}
  s_{\pi(n)}(\bm{\sigma})=s_n(\bm{\sigma}_\pi)
\end{equation}
holds for any $\bm{\sigma}\in\Sigma$ and for any $n\in\{1,\ldots,N\}$,
where $\bm{\sigma}_\pi\equiv(\sigma_{\pi(1)},\ldots,\sigma_{\pi(N)})$. 
\end{definition}
The following definition is due to Ref.~\cite{Plan2017}. 
\begin{definition}
A game is \emph{weakly symmetric} if for any pair of players $n$
and $\bar{n}$ there exists some permutation $\pi$ on $\{1,\ldots,N\}$
satisfying $\pi(n)=\bar{n}$ such that 
the game is symmetric with respect to $\pi$.
\end{definition}
Consider an $N$-player weakly symmetric game. 
Assume that one player simultaneously enforces
$N$ independent linear relations on the average payoffs
$\{e_n\}_{n\in\{1,\ldots,N\}}$ of $N$ players
via adopting an $N$-dimensional ZD strategy.
(Note that for this to be possible the number $M_n$ of actions
should satisfy $M_n\ge N+1$).
Then, the average payoffs $\{e_n\}_{n\in\{1,\ldots,N\}}$ 
should be simultaneously controlled,
but they should satisfy $e_1=e_2=\cdots=e_n$ 
due to the consistency of ZD strategies.

The difficulty of construction of a ZD strategy of one player with dimension $N$ in weakly symmetric $N$-player games can be seen in the following two propositions, whose proofs are given in Methods.
\begin{proposition}
\label{prop:nosimZD_nonzeroPD}
In a weakly symmetric $N$-player game, 
if the strategy vectors of one player contain no zero element,
then a ZD strategy of the player with dimension $N$ is impossible. 
\end{proposition}
\begin{proposition}
\label{prop:nosimZD_difpayoff}
In a weakly symmetric $N$-player game, if payoffs $s_n(\bm{\sigma})$ of player $n$ are different from each other for all $\bm{\sigma}$, then a ZD strategy with dimension $N$ is impossible.
\end{proposition}

\section*{Discussion}
In this paper, we have derived ZD strategies for general multi-player multi-action public-monitoring games, in which players cannot observe actions of other players.
By formulating ZD strategy in terms of linear algebra, we have proved that linear payoff relations enforced by ZD players are consistent.
Furthermore, we have proved that linear payoff relations enforced by players with ZD strategies are independent under a general condition. 
We emphasize that these results hold not only for imperfect-monitoring games but also for perfect-monitoring games.
We have also provided a simple example in which one player can simultaneously enforce more than one linear constraints on the expected payoffs. 
These results elucidate constraints on ZD strategies in terms of linear algebra.

Although we have discussed mathematical properties of ZD strategies if exist, we do not know the criterion for whether ZD strategies exist or not when a game is given.
For example, we can easily show that ZD strategy does not exist for the rock-paper-scissors game, which is the simplest two-player three-action symmetric zero-sum game.
(See \nameref{S1_Text} for the proof.)
Whereas, we can also show that there is a two-player three-action symmetric zero-sum game for which ZD strategy exists, which is also provided in \nameref{S1_Text}.
Generally, the dimension of $\spn\mathcal{S}$ is smaller than $N+1$ for zero-sum games, and construction of ZD strategies for zero-sum games is expected to be more difficult compared to non-zero-sum games.
Consistency together with constraints on payoffs such as symmetry and linear dependence may be useful to specify the space of ZD strategies which can exist.
Specifying a general criterion for the existence of ZD strategies is an important future problem.

In addition, it should be noted that ZD strategies are not always ``rational'' strategies, which have been a main subject of game theory.
Therefore, investigation of ZD strategies in terms of bounded rationality~\cite{Rub1998} may be needed.
Specifying the situation where ZD strategies are adopted is another important problem.

Another remark is related to memory of strategies.
In this work, we considered only memory-one strategies.
In Ref.~\cite{PreDys2012}, it has been proved that a player with longer memory does not have advantage over a player with short memory in terms of average payoff in two-player games.
In Ref.~\cite{PHRT2015,HDND2016}, it has been shown that this statement also holds for multi-player games.
Therefore, considering only memory-one strategies should be sufficient even in our public-monitoring situation.
Longer memory strategies attract much attentions in repeated games with implementation errors~\cite{HMCN2017,MurBae2018}.
Extension of ZD strategies to longer memory case may lead to different evolutionary behavior compared to memory-one strategies.

We remark on the effect of imperfect monitoring.
In perfect monitoring case, the strategy vectors are arbitrary as long as they satisfy the conditions for probability distributions.
In contrast, in imperfect monitoring case, forms of the strategy vectors are constrained by Eq.~\eqref{eq:transition_n}.
Therefore, the space of ZD strategies for imperfect-monitoring games is generally smaller than that for perfect-monitoring games.
In \nameref{S1_Text}, we provide examples of ZD strategies in simple imperfect-monitoring games.

\section*{Methods}
\subsection*{Proof of Lemma~\ref{lemma:consistent}}
Assume to the contrary that 
$\bm{v}\equiv\gamma \mathbf{1}_M\in\spn\mathcal{T}$ with $\gamma\not=0$.
Taking the inner product of $\bm{v}$ with 
the stationary distribution $\bm{\rho}$, one has 
$\bm{\rho}^\mathsf{T} \bm{v} = 0$ since $\bm{v}\in\spn\mathcal{T}$
is represented as a linear combination of the strategy vectors 
and since the inner product of a strategy vector and
the stationary distribution is zero. 
On the other hand, 
$\gamma \bm{\rho}^\mathsf{T} \mathbf{1}_M = \gamma$ holds 
because of the normalization of the stationary distribution.
Therefore we obtain $\gamma=0$, leading to contradiction.

\subsection*{Proof of Proposition~\ref{prop:Indep}}
We first show the following lemma.
\begin{lemma}
\label{lemma:indep}
Let $N'$ be a subset of players.
Assume that $\tilde{\bm{T}}_n(\sigma_n)$ does not have zero elements
for any $n\in N'$ and any $\sigma_n\in\{1,\ldots,M_n\}$.
For $n\in N'$, let $\bm{v}_n$ be an arbitrary non-zero vector
in $\spn\mathcal{T}_n$.
Then $\{\bm{v}_n\}_{n\in N'}$ are linearly independent.
\end{lemma}

\begin{proof}
We assume to the contrary that $\{\bm{v}_n\}_{n\in N'}$ are linearly dependent.
Then there is a set of coefficients $\{a_n\}_{n\in N'}$ with which
$\sum_{n\in N'}a_n\bm{v}_n=\mathbf{0}_M$ holds.
Without loss of generality we assume $a_n\not=0$ for $n\in N'$.

Since $\bm{v}_n\in\spn\mathcal{T}_n$, it is expressed as
$\bm{v}_n=\mathcal{T}_n\bm{c}_n$ with a non-zero vector
$\bm{c}_n=(c_{n,1},\ldots,c_{n,M_n})^\mathsf{T}$. 
Let $\tilde{\sigma}_n\equiv \arg\min_{\sigma_n\in\{1,\ldots,M_n\}}\{a_nc_{n,\sigma_n}\}$,
where ties may be broken arbitrarily,
and $\tilde{c}_n\equiv c_{n,\tilde{\sigma}_n}$.
With Eq.~\eqref{eq:null}, one obtains
\begin{equation}
  \bm{v}_n=\mathcal{T}_n(\bm{c}_n-\tilde{c}_n\mathbf{1}_{M_n}),
\end{equation}
and thus
\begin{equation}
  a_nv_n(\bm{\sigma}')=\sum_{\sigma_n=1}^{M_n}a_n(c_{n,\sigma_n}-\tilde{c}_n)
  \tilde{T}_n(\sigma_n|\bm{\sigma}').
\end{equation}

We show that the inequality 
\begin{equation}
  \label{eq:l2ineq}
  a_n(c_{n,\sigma_n}-\tilde{c}_n)\tilde{T}_n(\sigma_n|\bm{\sigma}')\ge0
\end{equation}
holds for any $n$, any $\sigma_n\in\{1,\ldots,M_n\}$,
and any $\bm{\sigma}'\in\Sigma$ satisfying $\sigma_n'=\tilde{\sigma}_n$. 
We first note that for any strategy vector $\tilde{\bm{T}}_n(\sigma_n)$
with action $\sigma_n\in\{1,\cdots,M_n\}$,
one has, from Eq.~\eqref{eq:def_tTn},
\begin{equation}
  \label{eq:orthant}
  \tilde{T}_n(\sigma_n|\bm{\sigma^\prime})
  \left\{
  \begin{array}{ll}
    \le0, & \sigma_n'=\sigma_n,\\
    \ge0, & \sigma_n'\not=\sigma_n.
  \end{array}
  \right.
\end{equation}
Fix any $\bm{\sigma}'\in\Sigma$
satisfying $\sigma_n'=\tilde{\sigma}_n$ for a moment. 
Then, for $\sigma_n=\tilde{\sigma}_n$ one has $c_{n,\sigma_n}=\tilde{c}_n$
by definition, making the left-hand side of Eq.~\eqref{eq:l2ineq}
equal to zero.
For $\sigma_n\not=\tilde{\sigma}_n$, on the other hand, 
one has $a_n(c_{n,\sigma_n}-\tilde{c}_n)\ge0$ by definition. 
Also, since $\sigma_n'=\tilde{\sigma}_n\not=\sigma_n$, from Eq.~\eqref{eq:orthant} 
one has $\tilde{T}_n(\sigma_n|\bm{\sigma}')\ge0$.
These imply that the inequality~\eqref{eq:l2ineq} holds
for $\sigma_n\not=\tilde{\sigma}_n$.
Putting the above arguments together, we have shown
that the inequality~\eqref{eq:l2ineq} holds 
for any $n$, any $\sigma_n\in\{1,\ldots,M_n\}$,
and any $\bm{\sigma}'\in\Sigma$ satisfying $\sigma_n'=\tilde{\sigma}_n$.

Fix any $\bm{\sigma}'\in\Sigma$ satisfying $\sigma_n'=\tilde{\sigma}_n$
for all $n\in N'$. 
The above argument has shown that
the inequality~\eqref{eq:l2ineq}
holds for any $n$ and any $\sigma_n\in\{1,\ldots,M_n\}$. 
On the other hand, at the beginning of the proof we have assumed that
\begin{equation}
  \sum_{n\in N'}a_nv_n(\bm{\sigma}')
  =\sum_{n\in N'}\sum_{\sigma_n=1}^{M_n}
  a_n(c_{n,\sigma_n}-\tilde{c}_n)\tilde{T}_n(\sigma_n|\bm{\sigma}')
  =0
\end{equation}
holds, 
implying that the summand $a_n(c_{n,\sigma_n}-\tilde{c}_n)\tilde{T}_n(\sigma_n|\bm{\sigma}')$ is equal to zero 
for any $n\in N'$ and any $\sigma_n\in\{1,\ldots,M_n\}$. 
By assumption, $a_n\not=0$ and $\tilde{T}_n(\sigma_n|\bm{\sigma}')\not=0$,
so that one has $c_{n,\sigma_n}=\tilde{c}_n$, and consequently,
$\bm{v}_n=\tilde{c}_n\mathcal{T}_n\mathbf{1}_{M_n}=\mathbf{0}_M$,
leading to contradiction.
\end{proof}

The proof of Proposition~\ref{prop:Indep} is straightforward
by taking $\bm{v}_n$ as belonging to $\mathcal{S}$ in Lemma~\ref{lemma:indep}.

\subsection*{Proof of Proposition~\ref{prop:nosimZD_nonzeroPD}}
We first show the following lemma.
\begin{lemma}
\label{lemma:sym}
Consider an $N$-player game which is symmetric with respect to
a permutation $\pi$ on $\{1,\ldots,N\}$.
Assume that the column vectors of $\mathcal{S}$ are linearly independent. 
For any pair of players $n$ and $\bar{n}$ satisfying $n\not=\pi(\bar{n})$, 
if the strategy vectors of these players contain
no zero element, 
then it is impossible for these players to adopt 
ZD strategies with which player $n$ enforces linear relation 
$\bm{e}^{\mathsf{T}}\bm{\alpha}=0$ with $\bm{\alpha}\not=\mathbf{0}_{N+1}$, 
and where player $\bar{n}$ enforces $\bm{e}^{\mathsf{T}}\bm{\alpha}_\pi=0$, 
where $\bm{\alpha}_\pi\equiv(\alpha_0,\alpha_{\pi(1)},\ldots,\alpha_{\pi(N)})^{\mathsf{T}}$. 
\end{lemma}

\begin{proof}
We assume to the contrary that there exists $\bm{\alpha}\not=\mathbf{0}_{N+1}$
satisfying the properties stated in Lemma~\ref{lemma:sym}. 
By assumption,
$\mathcal{S}\bm{\alpha}\in V_n=\spn\mathcal{T}_n\cap\spn\mathcal{S}$
and $\mathcal{S}\bm{\alpha}_\pi\in V_{\bar{n}}=\spn\mathcal{T}_{\bar{n}}\cap\spn\mathcal{S}$.
There then exist $\bm{c}_n$ and $\bar{\bm{c}}_{\bar{n}}$ satisfying
$\mathcal{T}_n\bm{c}_n=\mathcal{S}\bm{\alpha}$ 
and $\mathcal{T}_{\bar{n}}\bar{\bm{c}}_{\bar{n}}=\mathcal{S}\bm{\alpha}_\pi$.
One has
\begin{align}
  \label{eq:symm}
  (\mathcal{S}\bm{\alpha}_\pi)(\bm{\sigma}'_\pi)
  &=\alpha_0+\sum_{n=1}^N\alpha_{\pi(n)}s_n(\bm{\sigma}'_\pi)
  \nonumber\\
  &=\alpha_0+\sum_{n=1}^N\alpha_{\pi(n)}s_{\pi(n)}(\bm{\sigma}')
  =(\mathcal{S}\bm{\alpha})(\bm{\sigma}'),
\end{align}
where the second equality is due to the assumed symmetry of the game
with respect to $\pi$.
Letting $\tilde{T}_{\bar{n},\pi}(\sigma_{\bar{n}}|\bm{\sigma}')\equiv\tilde{T}_{\bar{n}}(\sigma_{\bar{n}}|\bm{\sigma}'_\pi)$,
$\tilde{\bm{T}}_{\bar{n},\pi}(\sigma_{\bar{n}})\equiv(\tilde{T}_{\bar{n},\pi}(\sigma_{\bar{n}}|\bm{\sigma}'))$,
and $\mathcal{T}_{\bar{n},\pi}\equiv(\tilde{\bm{T}}_{\bar{n},\pi}(1),\ldots,\tilde{\bm{T}}_{\bar{n},\pi}(M_{\bar{n}}))$,
one has
\begin{align}
  (\mathcal{T}_{\bar{n},\pi}\bar{\bm{c}}_{\bar{n}})(\bm{\sigma}')
  &=(\mathcal{T}_{\bar{n}}\bar{\bm{c}}_{\bar{n}})(\bm{\sigma}'_\pi)
  =(\mathcal{S}\bm{\alpha}_\pi)(\bm{\sigma}'_\pi)
  \nonumber\\
  &=(\mathcal{S}\bm{\alpha})(\bm{\sigma}')=(\mathcal{T}_n\bm{c}_n)(\bm{\sigma}'),
\end{align}
implying that $\mathcal{T}_{\bar{n},\pi}\bar{\bm{c}}_{\bar{n}}=\mathcal{T}_n\bm{c}_n$ holds.
Let $\bm{v}=\mathcal{T}_n\bm{c}_n=\mathcal{T}_{\bar{n},\pi}\bar{\bm{c}}_{\bar{n}}$. 

Let $\sigma_{n,\mathrm{max}}=\arg\max_{\sigma_n}c_{n,\sigma_n}$ and 
$\bar{\sigma}_{\bar{n},\mathrm{min}}=\arg\min_{\sigma_{\bar{n}}}\bar{c}_{\bar{n},\sigma_{\bar{n}}}$,
where ties may be broken arbitrarily,
and $c_{n,\mathrm{max}}=c_{n,\sigma_{n,\mathrm{max}}}$ and
$\bar{c}_{\bar{n},\mathrm{min}}=\bar{c}_{\bar{n},\bar{\sigma}_{n,\mathrm{min}}}$.
One then has
\begin{equation}
  \bm{v}=\mathcal{T}_n(\bm{c}_n-c_{n,\mathrm{max}}\mathbf{1}_{M_n})
  =\mathcal{T}_{\bar{n},\pi}(\bar{\bm{c}}_{\bar{n}}-\bar{c}_{\bar{n},\mathrm{min}}\mathbf{1}_{M_{\bar{n}}}).
\end{equation}
Recalling that we have assumed $n\not=\pi(\bar{n})$,
let $\bm{\sigma}'\in\Sigma$ be an arbitrary state
satisfying $\sigma'_n=\sigma_{n,\mathrm{max}}$
and $\sigma'_{\pi(\bar{n})}=\bar{\sigma}_{\bar{n},\mathrm{min}}$.
Then, in view of Eq.~\eqref{eq:orthant}, one has
\begin{align}
  v(\bm{\sigma}')
  &=\sum_{\sigma_n=1}^{M_n}(c_{n,\sigma_n}-c_{n,\mathrm{max}})\tilde{T}_n(\sigma_n|\bm{\sigma}')\le0,
  \nonumber\\
  &=\sum_{\sigma_{\bar{n}}=1}^{M_{\bar{n}}}(\bar{c}_{\bar{n},\sigma_{\bar{n}}}-\bar{c}_{\bar{n},\mathrm{min}})
  \tilde{T}_{\bar{n}}(\sigma_{\bar{n}}|\bm{\sigma}'_\pi)\ge0,
\end{align}
implying that $v(\bm{\sigma}')=0$ holds.
Since $(c_{n,\sigma_n}-c_{n,\mathrm{max}})\tilde{T}_n(\sigma_n|\bm{\sigma}')\le0$
for all $\sigma_n\in\{1,\ldots,M_n\}$, they are all equal to zero.
Since $\tilde{T}_n(\sigma_n|\bm{\sigma}')$ is assumed non-zero,
one has $c_{n,\sigma_n}=c_{n,\mathrm{max}}$ for all $\sigma_n\in\{1,\ldots,M_n\}$
and consequently $\bm{c}_n\propto\mathbf{1}_{M_n}$.
One similarly has $\bar{\bm{c}}_{\bar{n}}\propto\mathbf{1}_{M_{\bar{n}}}$.
Therefore, from Eq.~\eqref{eq:null} one has
$\mathcal{T}_n\bm{c}_n=\mathcal{T}_{\bar{n}}\bar{\bm{c}}_{\bar{n}}=\mathbf{0}_M$.
Due to the assumption of linear independence of the columns of $\mathcal{S}$,
it in turn implies that $\bm{\alpha}=\mathbf{0}_{N+1}$ holds,
leading to contradiction. 
\end{proof}

It should be noted that Lemma~\ref{lemma:sym} holds
even if one takes $\bar{n}=n$, in which case the Lemma implies
that, if the game is symmetric with respect to $\pi$, 
player $n$ with $\pi(n)\not=n$ cannot enforce
linear relations $\bm{e}^{\mathsf{T}}\bm{\alpha}=\bm{e}^{\mathsf{T}}\bm{\alpha}_\pi=0$ simultaneously. 
It should also be noted that Lemma~\ref{lemma:sym} furthermore implies
that it is impossible for that player to enforce
a linear relation $\bm{e}^{\mathsf{T}}\bm{\alpha}=0$
satisfying $\bm{\alpha}_\pi=\bm{\alpha}\not=\mathbf{0}_{N+1}$.
In other words, in a symmetric game no player to whom the game is symmetric
can enforce a linear relation with the same symmetry as the game itself.

Proposition~\ref{prop:nosimZD_nonzeroPD} is a direct consequence of Lemma~\ref{lemma:sym} in weakly symmetric multi-player games.

\subsection*{Proof of Proposition~\ref{prop:nosimZD_difpayoff}}
Without loss of generality, we assume that player $k$ takes an $N$-dimensional ZD strategy
determining the average payoffs $e_n$ for $n=1,\cdots, N$.
Due to the above discussion, only $e_1=\cdots=e_N=C$ is allowed.
Letting $\bm{\alpha}^{(n)}\equiv(-C,0,\cdots,\mathop{1}\limits_{\widehat n},\cdots,0)^{\mathsf{T}}$ for $n\in\{1,\ldots,N\}$,
one can take $\{\mathcal{S}\bm{\alpha}^{(n)}\}_{n\in\{1,\ldots,N\}}$ as a basis
of the $N$-dimensional ZD strategy.
Let $\bm{c}^{(n)}$ be defined as 
\begin{equation}
  \mathcal{T}_k\bm{c}^{(n)}=\mathcal{S}\bm{\alpha}^{(n)}
  =\bm{s}_n-C\mathbf{1}_M
  ,\quad n\in\{1,\ldots,N\}.
\end{equation}
By the assumption of weak symmetry,
for any player $n\not=k$, 
there exists a permutation $\pi$ satisfying $\pi(n)=k$ such that
the game is symmetric with respect to $\pi$.
Noting that $\bm{\alpha}^{(n)}_\pi=\bm{\alpha}^{(k)}$,
from Eq.~\eqref{eq:symm} one has
\begin{equation}
  \label{eq:simultaneous_sym}
  (\mathcal{T}_k\bm{c}^{(k)})(\bm{\sigma}'_\pi)
  =(\mathcal{T}_k\bm{c}^{(n)})(\bm{\sigma}').
\end{equation}

For $n\in\{1,\ldots,N\}$,
define $\sigma_{\mathrm{max}}^{(n)}\equiv\arg\max_{\sigma_k}c_{\sigma_k}^{(n)}$
and $\sigma_{\mathrm{min}}^{(n)}\equiv\arg\min_{\sigma_k}c_{\sigma_k}^{(n)}$,
where ties may be broken arbitrarily provided that
$\sigma_{\mathrm{max}}^{(n)}\not=\sigma_{\mathrm{min}}^{(n)}$ holds, 
and $c^{(n)}_{\mathrm{max}} = c^{(n)}_{\sigma_{\mathrm{max}}}$
and $c^{(n)}_{\mathrm{min}} = c^{(n)}_{\sigma_{\mathrm{min}}}$.
From Eq.~\eqref{eq:tt_norm}, one has
\begin{align}
  \mathcal{T}_k\bm{c}^{(n)}
  &=\mathcal{T}_k\bigl(\bm{c}^{(n)}-c_{\mathrm{max}}^{(n)}\mathbf{1}_{M_k}\bigr)
  \nonumber\\
  &=\mathcal{T}_k\bigl(\bm{c}^{(n)}-c_{\mathrm{min}}^{(n)}\mathbf{1}_{M_k}\bigr).
\end{align}
Then, from Eq.~\eqref{eq:simultaneous_sym} and Eq.~\eqref{eq:orthant}, we obtain for an arbitrary $\bm{\sigma}^*$ satisfying $\sigma^*_{k}=\sigma^{(n)}_{\mathrm{max}}$ and $\sigma^*_{n}=\sigma^{(k)}_{\mathrm{min}}$
\begin{align}
 s_n\left( \bm{\sigma}^* \right) - C
 &= \left( \mathcal{T}_k \bigl( \bm{c}^{(n)} - c^{(n)}_{\mathrm{max}} \bm{1}_{M_k} \bigr) \right) \left( \bm{\sigma}^* \right) \leq 0 \\
 &= \left( \mathcal{T}_k \bigl( \bm{c}^{(k)} - c^{(k)}_{\mathrm{min}} \bm{1}_{M_k} \bigr) \right) \left( \bm{\sigma}^*_{\pi} \right) \geq 0
\end{align}
implying $s_n\left( \bm{\sigma}^* \right) = C$.
On the other hand, we also obtain for an arbitrary $\bm{\sigma}^{**}$ satisfying $\sigma^{**}_{k}=\sigma^{(n)}_{\mathrm{min}}$ and $\sigma^{**}_{n}=\sigma^{(k)}_{\mathrm{max}}$
\begin{align}
 s_n\left( \bm{\sigma}^{**} \right) - C
 &= \left( \mathcal{T}_k \bigl( \bm{c}^{(n)} - c^{(n)}_{\mathrm{min}} \bm{1}_{M_k} \bigr) \right) \left( \bm{\sigma}^{**} \right) \geq 0 \\
 &= \left( \mathcal{T}_k \bigl( \bm{c}^{(k)} - c^{(k)}_{\mathrm{max}} \bm{1}_{M_k} \bigr) \right) \left( \bm{\sigma}^{**}_{\pi} \right) \leq 0
\end{align}
implying $s_n\left( \bm{\sigma}^{**} \right) = C$.
Then, because we have assumed that all elements of the payoff vector $\bm{s}_n$ are different from each other,
we have arrived at a contradiction.

\clearpage
\section*{Supporting information}
\label{S1_Text}
\setcounter{equation}{0}
\def\theequation{S\arabic{equation}}
\setcounter{figure}{0}
\def\thefigure{S\arabic{figure}}

\section{ZD strategy in zero-sum games}
\subsection{Absence of ZD strategies in the rock-paper-scissors game}
We consider the rock-paper-scissors game:
\begin{align}
 \bm{s}_1 &= \left( 0, 1, -1, -1, 0, 1, 1, -1, 0 \right)^\mathsf{T} \nonumber \\
 \bm{s}_2 &= \left( 0, -1, 1, 1, 0, -1, -1, 1, 0 \right)^\mathsf{T}.
 \label{eq:payoff_RPS}
\end{align}
It should be noted that this game is two-player three-action symmetric zero-sum game.
Because it is zero-sum game, the payoff vectors are linearly dependent $\bm{s}_2=-\bm{s}_1$, and
\begin{eqnarray}
 \left\langle s_2 \right\rangle_\mathrm{s} &=& - \left\langle s_1 \right\rangle_\mathrm{s}
\end{eqnarray}
always holds.
When player $1$ can employ ZD strategy, her strategy takes the form
\begin{eqnarray}
 \sum_{\sigma_1=1}^{3} c_{\sigma_1}^{(1)} \tilde{\bm{T}}_1(\sigma_1) &=& \alpha \bm{s}_1 + \gamma \bm{1}_9
 \label{eq:ZD_1}
\end{eqnarray}
and she enforces the linear relation
\begin{eqnarray}
 0 &=& \alpha \left\langle s_1 \right\rangle_\mathrm{s} + \gamma.
\end{eqnarray}
Since the game is symmetric, player $2$ can independently employ ZD strategy which enforces
\begin{eqnarray}
 0 &=& \alpha \left\langle s_2 \right\rangle_\mathrm{s} + \gamma.
\end{eqnarray}
Then, because of the consistency of ZD strategies, $\gamma=0$ must hold.
On the other hand, by using $\gamma=0$, Eq.~\eqref{eq:ZD_1} can be written as
\begin{eqnarray}
 \alpha s_1\left( \sigma_1^\prime, \sigma_2^\prime \right) &=& \sum_{\sigma_1=1}^{3} \left( c_{\sigma_1}^{(1)} - c_\mathrm{max}^{(1)} \right) \tilde{T}_1\left(\sigma_1 | \sigma_1^\prime, \sigma_2^\prime \right) \qquad \left( \forall \sigma_1^\prime, \forall \sigma_2^\prime \right) \\
 &=& \sum_{\sigma_1=1}^{3} \left( c_{\sigma_1}^{(1)} - c_\mathrm{min}^{(1)} \right) \tilde{T}_1\left(\sigma_1 | \sigma_1^\prime, \sigma_2^\prime \right) \qquad \left( \forall \sigma_1^\prime, \forall \sigma_2^\prime \right),
\end{eqnarray}
where we have defined
\begin{eqnarray}
 c_\mathrm{max}^{(1)} &\equiv& \max_{\sigma_1} c_{\sigma_1}^{(1)} \\
 c_\mathrm{min}^{(1)} &\equiv& \min_{\sigma_1} c_{\sigma_1}^{(1)} \\
 \sigma_{1, \mathrm{max}} &\equiv& \arg\max_{\sigma_1} c_{\sigma_1}^{(1)} \\
 \sigma_{1, \mathrm{min}} &\equiv& \arg\min_{\sigma_1} c_{\sigma_1}^{(1)}
\end{eqnarray}
Then we find that
\begin{eqnarray}
 \alpha s_1\left( \sigma_{1, \mathrm{max}}, \sigma_2^\prime \right) &=& \sum_{\sigma_1=1}^{3} \left( c_{\sigma_1}^{(1)} - c_\mathrm{max}^{(1)} \right) \tilde{T}_1\left(\sigma_1 | \sigma_{1, \mathrm{max}}, \sigma_2^\prime \right) \leq 0 \qquad \left( \forall \sigma_2^\prime \right)
 \label{eq:payoff_1_max}
\end{eqnarray}
and
\begin{eqnarray}
 \alpha s_1\left( \sigma_{1, \mathrm{min}}, \sigma_2^\prime \right) &=& \sum_{\sigma_1=1}^{3} \left( c_{\sigma_1}^{(1)} - c_\mathrm{min}^{(1)} \right) \tilde{T}_1\left(\sigma_1 | \sigma_{1, \mathrm{min}}, \sigma_2^\prime \right) \geq 0 \qquad \left( \forall \sigma_2^\prime \right).
 \label{eq:payoff_1_min}
\end{eqnarray}
Because $\sigma_{1, \mathrm{max}}\in \left\{ 1, 2, 3 \right\}$ and $\sigma_{1, \mathrm{min}}\in \left\{ 1, 2, 3 \right\}$, Eq.~\eqref{eq:payoff_1_max} and Eq.~\eqref{eq:payoff_1_min} are inconsistent with the definition of the payoff~\eqref{eq:payoff_RPS}.
Therefore, we conclude that ZD strategy does not exist in the rock-paper-scissors game.

\subsection{Example of ZD strategy in two-player three-action symmetric zero-sum game}
We next consider the following two-player three-action symmetric zero-sum game, which is the slightly modified version of the game in the main text:
\begin{eqnarray}
 \bm{s}_1 &=& \left( 0, r, 0, -r, 0, 0, 0, 0, 0 \right)^\mathsf{T} \nonumber \\
 \bm{s}_2 &=& \left( 0, -r, 0, r, 0, 0, 0, 0, 0 \right)^\mathsf{T}.
\end{eqnarray}
We remark that $\bm{s}_1$ and $\bm{s}_2$ are linearly dependent $\bm{s}_2=-\bm{s}_1$.
We choose strategies of player $1$ as
\begin{eqnarray}
 \bm{T}_1(1) &=& \left( 1, 1-p, 1, p^\prime, 0, 0, 0, 0, 0 \right)^\mathsf{T} \nonumber \\
 \bm{T}_1(2) &=& \left( 0, q, 0, 1-q^\prime, 1, 1, 0, 0, 0 \right)^\mathsf{T} \nonumber \\
 \bm{T}_1(3) &=& \left( 0, p-q, 0, q^\prime-p^\prime, 0, 0, 1, 1, 1 \right)^\mathsf{T}
\end{eqnarray}
with $0\leq p\leq 1$, $0\leq q\leq 1$, $0\leq p^\prime \leq 1$, $0\leq q^\prime \leq 1$, $q\leq p$, and $p^\prime \leq q^\prime$.
Then we obtain
\begin{eqnarray}
 r\frac{q^\prime - q}{p^\prime q - pq^\prime} \bm{\tilde{T}}_1(1) + r\frac{p^\prime - p}{p^\prime q - pq^\prime} \bm{\tilde{T}}_1(2) &=& \bm{s}_1.
\end{eqnarray}
Therefore, this strategy is ZD strategy which control the payoffs of both players as $\left\langle s_1 \right\rangle_\mathrm{s}=\left\langle s_2 \right\rangle_\mathrm{s}=0$.
It should be noted that Eq.~\eqref{eq:payoff_1_max} and Eq.~\eqref{eq:payoff_1_min} are satisfied for this game.

\section{ZD strategy in game with public monitoring}
\subsection{Two-player two-action game}
As an example of ZD strategy for a repeated imperfect-monitoring game,
we consider a two-player two-action symmetric game~\cite{Kob2018}.
We assume $\tau\in \left\{ 1, 2 \right\}$ and the probability $W\left( \tau | \bm{\sigma}' \right)$ is given by
\begin{eqnarray}
 W\left( 1 | 1,1 \right) &=& \frac{1}{2}, \\
 W\left( 1 | 1,2 \right) &=& w, \\
 W\left( 1 | 2,1 \right) &=& 1-w, \\
 W\left( 1 | 2,2 \right) &=& \frac{1}{2}.
\end{eqnarray}
This model is different from the noisy games studied by Hao et al.~\cite{HRZ2015},
in that they consider $\tau$ as noisy states, taking four values $\tau\in \left\{ gg, gb, bg, bb \right\}$ corresponding to the four states in the iterated prisoner's dilemma game, whereas ours considers $\tau$ as taking only two values,
representing winning/losing of player 1.
The payoff vectors are given by $\textrm{\boldmath $s$}_1 = \left( R, S, T, P \right)^\mathsf{T}$ and $\textrm{\boldmath $s$}_2 = \left( R, T, S, P \right)^\mathsf{T}$.
We consider equalizer strategy for player $1$:
\begin{eqnarray}
 \tilde{\bm{T}}_1(1) &=& \beta \textrm{\boldmath $s$}_2 + \gamma \textrm{\boldmath $1$}_4.
 \label{eq:equalizer}
\end{eqnarray}
This strategy unilaterally sets the average payoff of player $2$ in the steady state:
\begin{eqnarray}
 \left\langle s_2 \right\rangle_\mathrm{s} &=& - \frac{\gamma}{\beta}.
\end{eqnarray}
By solving Eq. (\ref{eq:equalizer}) with respect to $\hat{T}_n\left( \sigma_n | \sigma^\prime_n, \tau \right)$, we obtain
\begin{eqnarray}
 \hat{T}_1\left( 1 | 1, 1 \right) &=& \frac{2(1-w)R-T}{1-2w}\beta + \gamma + 1, \\
 \hat{T}_1\left( 1 | 1, 2 \right) &=& \frac{T - 2wR}{1-2w}\beta + \gamma + 1, \\
 \hat{T}_1\left( 1 | 2, 1 \right) &=& \frac{S - 2wP}{1-2w}\beta + \gamma, \\
 \hat{T}_1\left( 1 | 2, 2 \right) &=& \frac{2(1-w)P - S}{1-2w}\beta + \gamma.
\end{eqnarray}

Concretely, we consider $w=1/5$ and $\left( R, S, T, P \right) = \left( 4, 1, 9/2, 3/2 \right)$.
By setting $\beta = -3/125$ and $\gamma=33/500$, we obtain 
\begin{eqnarray}
 \hat{T}_1\left( 1 | 1, 1 \right) &=& \frac{99}{100}, \\
 \hat{T}_1\left( 1 | 1, 2 \right) &=& \frac{95}{100}, \\
 \hat{T}_1\left( 1 | 2, 1 \right) &=& \frac{5}{100}, \\
 \hat{T}_1\left( 1 | 2, 2 \right) &=& \frac{1}{100},
\end{eqnarray}
and
\begin{eqnarray}
 \left\langle s_2 \right\rangle_\mathrm{s} &=& \frac{11}{4}.
 \label{eq:avs2_example}
\end{eqnarray}
In Fig.~\ref{fig:ZDS_eq}, we display the result of numerical simulation of one sample.
\begin{figure}[tbp]
\includegraphics[clip, width=8.0cm]{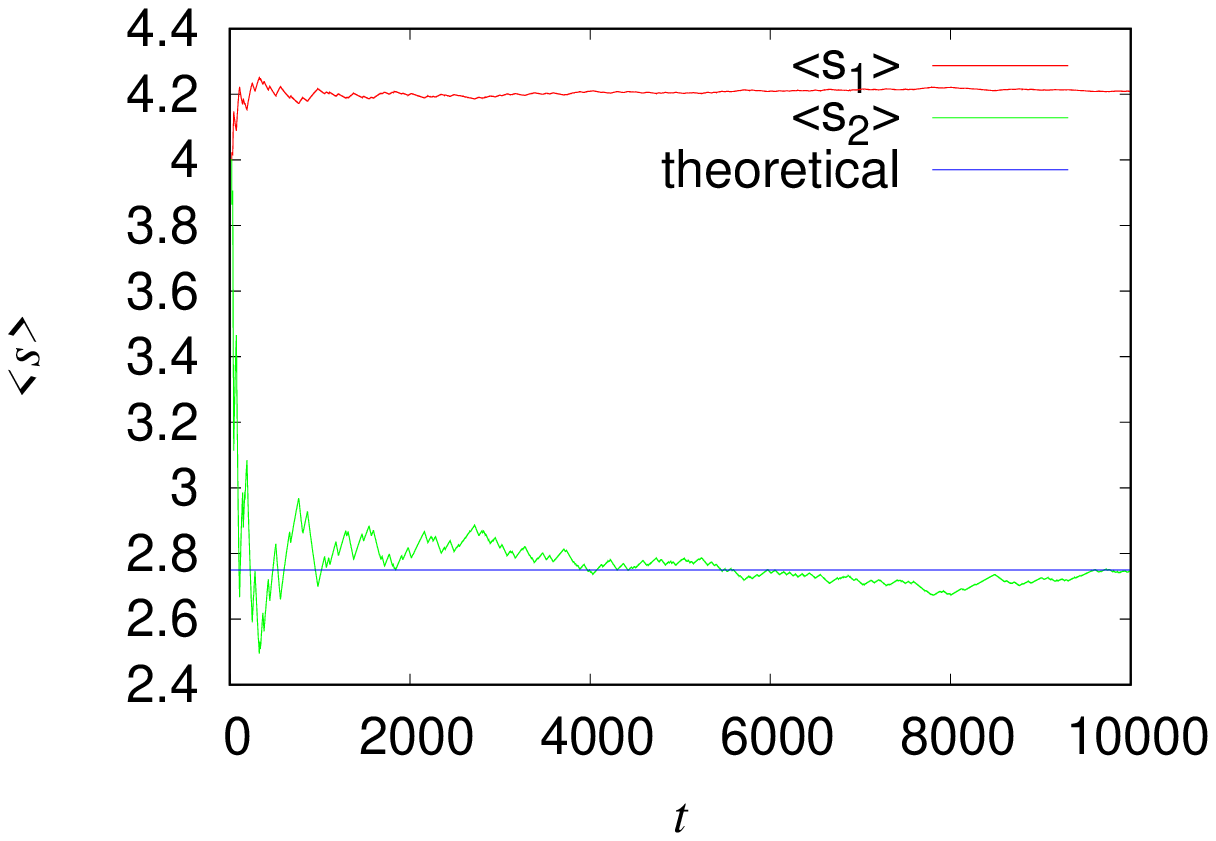}
\caption{Time-averaged payoffs of two players $\sum_{t^\prime=1}^t s_n\left( \sigma_1(t^\prime), \sigma_2(t^\prime) \right)/t$. The solid line corresponds to the theoretical prediction Eq. (\ref{eq:avs2_example}) for player $2$'s expected payoff.}
\label{fig:ZDS_eq}
\end{figure}
Time-averaged payoffs of two players $\sum_{t^\prime=1}^t s_n\left( \sigma_1(t^\prime), \sigma_2(t^\prime) \right)/t$ are displayed when the strategy of player $2$ is all-$1$:
\begin{eqnarray}
 \hat{T}_2\left( 1 | 1, 1 \right) &=& 1 \\
 \hat{T}_2\left( 1 | 1, 2 \right) &=& 1 \\
 \hat{T}_2\left( 1 | 2, 1 \right) &=& 1 \\
 \hat{T}_2\left( 1 | 2, 2 \right) &=& 1.
\end{eqnarray}
The initial condition is set to $\sigma_1(0)=1$ and $\sigma_2(0)=1$.
The numerical result for player $2$ well matches with the theoretical prediction Eq. (\ref{eq:avs2_example}).
We can see that the expected payoff of player $2$ is unilaterally controlled by the ZD strategy of player $1$.

\subsection{Two-player three-action game}
We consider the same two-player three-action symmetric game in the main text:
\begin{eqnarray}
 \bm{s}_1 &=& \left( 0, r_1, 0, r_2, 0, 0, 0, 0, 0 \right)^\mathsf{T} \nonumber \\
 \bm{s}_2 &=& \left( 0, r_2, 0, r_1, 0, 0, 0, 0, 0 \right)^\mathsf{T}.
\end{eqnarray}
However, we assume that players can observe only common information $\tau\in \left\{ \mathrm{y}, \mathrm{n} \right\}$ and the probability $W\left( \tau | \bm{\sigma}' \right)$ is given by
\begin{eqnarray}
 W\left( \mathrm{y} | 1,1 \right) &=& 0 \\
 W\left( \mathrm{y} | 1,2 \right) &=& w \\
 W\left( \mathrm{y} | 1,3 \right) &=& 0 \\
 W\left( \mathrm{y} | 2,1 \right) &=& w \\
 W\left( \mathrm{y} | 2,2 \right) &=& 0 \\
 W\left( \mathrm{y} | 2,3 \right) &=& 0 \\
 W\left( \mathrm{y} | 3,1 \right) &=& 0 \\
 W\left( \mathrm{y} | 3,2 \right) &=& 0 \\
 W\left( \mathrm{y} | 3,3 \right) &=& 0.
\end{eqnarray}
The common information $\tau$ represents whether payoffs of both players are non-zero or not.
We consider the situation that player $1$ employs the following strategy:
\begin{eqnarray}
 \hat{T}_1\left( 1 | 1, \mathrm{y} \right) &=& \frac{w-p}{w} \\
 \hat{T}_1\left( 1 | 1, \mathrm{n} \right) &=& 1 \\
 \hat{T}_1\left( 1 | 2, \mathrm{y} \right) &=& \frac{p^\prime}{w} \\
 \hat{T}_1\left( 1 | 2, \mathrm{n} \right) &=& 0 \\
 \hat{T}_1\left( 1 | 3, \mathrm{y} \right) &=& 0 \\
 \hat{T}_1\left( 1 | 3, \mathrm{n} \right) &=& 0 \\
 \hat{T}_1\left( 2 | 1, \mathrm{y} \right) &=& \frac{q}{w} \\
 \hat{T}_1\left( 2 | 1, \mathrm{n} \right) &=& 0 \\
 \hat{T}_1\left( 2 | 2, \mathrm{y} \right) &=& \frac{w-q^\prime}{w} \\
 \hat{T}_1\left( 2 | 2, \mathrm{n} \right) &=& 1 \\
 \hat{T}_1\left( 2 | 3, \mathrm{y} \right) &=& 0 \\
 \hat{T}_1\left( 2 | 3, \mathrm{n} \right) &=& 0 \\
 \hat{T}_1\left( 3 | 1, \mathrm{y} \right) &=& \frac{p-q}{w} \\
 \hat{T}_1\left( 3 | 1, \mathrm{n} \right) &=& 0 \\
 \hat{T}_1\left( 3 | 2, \mathrm{y} \right) &=& \frac{q^\prime - p^\prime}{w} \\
 \hat{T}_1\left( 3 | 2, \mathrm{n} \right) &=& 0 \\
 \hat{T}_1\left( 3 | 3, \mathrm{y} \right) &=& 1 \\
 \hat{T}_1\left( 3 | 3, \mathrm{n} \right) &=& 1
\end{eqnarray}
with $0\leq p\leq w$, $0\leq q\leq w$, $0\leq p^\prime \leq w$, $0\leq q^\prime \leq w$, $q\leq p$, and $p^\prime \leq q^\prime$.
Then, from the definition
\begin{eqnarray}
 T_1\left( \sigma_1 | \bm{\sigma}' \right) \equiv \sum_{\tau=\mathrm{y},\mathrm{n}} W\left( \tau | \bm{\sigma}' \right) \hat{T}_1\left( \sigma_1 | \sigma^\prime_1, \tau \right),
\end{eqnarray}
we obtain
\begin{eqnarray}
 \bm{T}_1(1) &=& \left( 1, 1-p, 1, p^\prime, 0, 0, 0, 0, 0 \right)^\mathsf{T} \nonumber \\
 \bm{T}_1(2) &=& \left( 0, q, 0, 1-q^\prime, 1, 1, 0, 0, 0 \right)^\mathsf{T} \nonumber \\
 \bm{T}_1(3) &=& \left( 0, p-q, 0, q^\prime-p^\prime, 0, 0, 1, 1, 1 \right)^\mathsf{T}.
\end{eqnarray}
These strategy vectors are the same as those in the main text, and give
\begin{eqnarray}
 \frac{q^\prime r_1 + qr_2}{p^\prime q - pq^\prime} \bm{\tilde{T}}_1(1) + \frac{p^\prime r_1 + pr_2}{p^\prime q - pq^\prime} \bm{\tilde{T}}_1(2) &= &\bm{s}_1 \\
 \frac{q^\prime r_2 + qr_1}{p^\prime q - pq^\prime} \bm{\tilde{T}}_1(1) + \frac{p^\prime r_2 + pr_1}{p^\prime q - pq^\prime} \bm{\tilde{T}}_1(2) &= &\bm{s}_2,
\end{eqnarray}
which enforce the linear relations $\left\langle s_1 \right\rangle_\mathrm{s}=0$ and $\left\langle s_2 \right\rangle_\mathrm{s}=0$.
Therefore, player $1$ can enforce the same linear relations as those in the perfect monitoring case, even though players can know only $\tau$.
This means that the space of states $\Sigma$ is successfully reduced to the smaller space $\left\{ \mathrm{y}, \mathrm{n} \right\}$ in terms of ZD strategies.
The only difference is that the possible region of parameters $p$, $q$, $p^\prime$, $q^\prime$ is smaller for $w\neq 1$ than that in perfect monitoring case.

In Fig.~\ref{fig:ZDS_two_three}, we display the result of numerical simulation of one sample for $r_1=2.0$ and $r_2=1.0$.
\begin{figure}[tbp]
\includegraphics[clip, width=8.0cm]{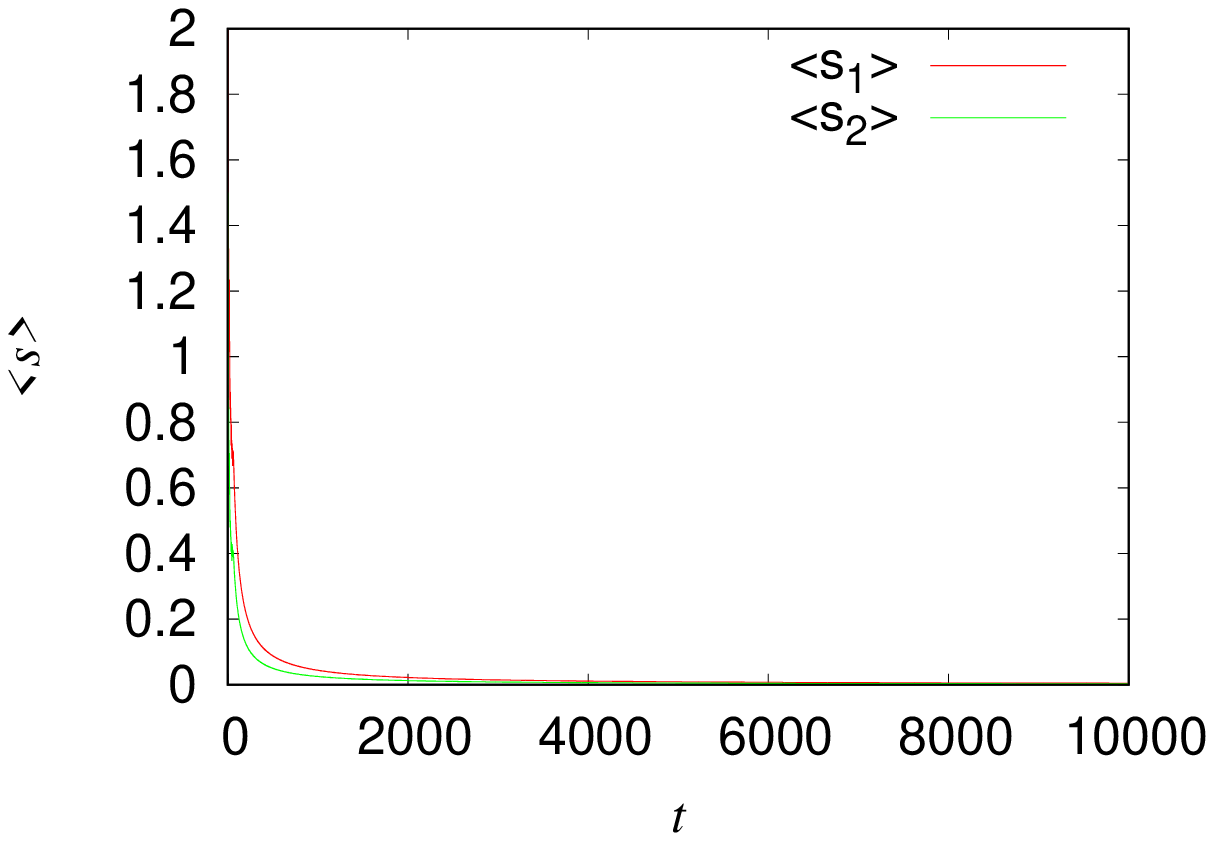}
\caption{Time-averaged payoffs of two players $\sum_{t^\prime=1}^t s_n\left( \sigma_1(t^\prime), \sigma_2(t^\prime) \right)/t$.}
\label{fig:ZDS_two_three}
\end{figure}
Time-averaged payoffs of two players $\sum_{t^\prime=1}^t s_n\left( \sigma_1(t^\prime), \sigma_2(t^\prime) \right)/t$ are displayed when $w=0.9$, $p=0.2$, $q=0.1$, $p^\prime=0.25$, $q^\prime=0.3$ and the strategy of player $2$ is
\begin{eqnarray}
 \hat{T}_2\left( \sigma_2 | \sigma_2^\prime, \tau \right) &=& \frac{1}{3} \qquad \left( \forall \sigma_2, \forall \sigma_2^\prime, \forall \tau \right).
\end{eqnarray}
The initial condition is given by the probability distribution $P(\sigma)=1/3$ for both players.
The numerical result is consistent with the theoretical prediction $\left\langle s_1 \right\rangle_\mathrm{s}=\left\langle s_2 \right\rangle_\mathrm{s}=0$.

\section*{Acknowledgments}
We thank Ryosuke Kobayashi for valuable discussions. This study was supported by JSPS KAKENHI Grant Numbers JP18H06476 and JP19K21542.

\nolinenumbers

% Either type in your references using
% \begin{thebibliography}{}
% \bibitem{}
% Text
% \end{thebibliography}
%
% or
%
% Compile your BiBTeX database using our plos2015.bst
% style file and paste the contents of your .bbl file
% here. See http://journals.plos.org/plosone/s/latex for 
% step-by-step instructions.
% 
%\begin{thebibliography}{10}
%\end{thebibliography}
\bibliography{zds}

\end{document}